%% file: OneModConj.tex
\documentclass[10pt,journal]{IEEEtran}
\usepackage{amssymb}
\usepackage{amsthm}
\usepackage{hyperref}
\usepackage{graphicx}
\usepackage{amsmath}
\usepackage{cite}
\usepackage{color}
\usepackage{subfig}
\numberwithin{equation}{section}

\newtheorem{theorem}{Theorem}

\newtheorem{conjecture}{Conjecture}

\newtheorem{lemma}{Lemma}

\newcommand{\ket}[1]{\left\vert#1\right\rangle}
\newcommand{\bra}[1]{\left\langle#1\right\vert}

\newcommand{\A}[1]{#1_A}

\newcommand{\beq}{\begin{equation}}
\newcommand{\eeq}{\end{equation}}
\newcommand{\baq}{\begin{eqnarray}}
\newcommand{\eaq}{\end{eqnarray}}

\newcommand{\brac}[1]{\lbrace #1 \rbrace}
\newcommand{\p}{\mathtt{p}}

\def\ket#1{| #1 \rangle}
\def\bra#1{\langle #1 |}

\def\A{\mathcal{A}}

\def\D{\mathcal{D}}
\def\E{\mathcal{E}}

\def\L{\mathcal{L}}
\def\M{\mathcal{M}}

\begin{document}
	
	\title{On the minimum output entropy of single-mode phase-insensitive Gaussian channels}
	\author{Haoyu~Qi, Mark~M.~Wilde, and~Saikat~Guha
\thanks{Haoyu Qi (hqi2@lsu.edu) and Mark M. Wilde (mwilde@lsu.edu) are with the Hearne Institute for Theoretical Physics, Department of Physics and Astronomy, Louisiana State University,  Baton Rouge, Louisiana 70803, USA. Mark M.~Wilde is also with the Center for Computation and Technology at LSU. Saikat Guha (sguha@bbn.com) is with the Quantum Information Processing (QuIP) group at Raytheon BBN Technologies.}}

	\maketitle
	\begin{abstract}
\input{abstract}
	\end{abstract}

\section{Introduction}\label{sec:intro}
\input{intro}

	\section{Background}\label{sec:background}
\input{background.tex}

	\section{Set up and main results}
	\label{sec:setup}
\input{mainresults}
	
	\section{An equivalent infinitesimal form}
	\label{sec:Th1}
\input{reduction}
	
\section{Proof attempt for Conjecture \ref{th:infinite-conjecture}}
\label{sec:th2}
\input{mainproof}

\section{Extending the proof to single-mode gauge-contravariant BGCs}\label{sec:contra}
\input{contravariant}

\section{Conclusion}\label{sec:conclusions}
\input{conclusions}

\section*{acknowledgements}
We are very grateful to Giacomo de Palma for pointing out the main problem with our previously claimed proof of Conjecture~\ref{thm:mainresult} in the original arXiv posting of this note.

This work was performed while Haoyu Qi was visiting the Quantum Information Processing (QuIP) group at Raytheon BBN Technologies, during the summer of 2016. HQ is supported by the Air Force Office of Scientific Research, the Army Research Office and the National Science Foundation. HQ thanks Jonathan P.~Dowling and Hwang Lee for their help and support of his visiting BBN. SG thanks Jeffrey H. Shapiro, MIT, for valuable discussions.\footnote{Shapiro attempted a direct proof of Conjecture~\ref{th:finite-conjecture} for a specific gauge-covariant BGC, the pure-loss channel, by writing down a Lagrange-multiplier optimization of the output state for a Fock-diagonal input state, from which he saw numerical evidence that the thermal input achieves the minimum output entropy~\cite{jeff_personal}. This special case (pure-loss channel) however was proved subsequently in~\cite{de2016gaussian}.} MMW acknowledges support from the NSF under Award No.~1350397.

This document does not contain technology or technical data controlled under either the U.S.~International Traffic in Arms Regulations or the U.S.~Export Administration Regulations.

\bibliographystyle{IEEEtran}
\bibliography{IEEEabrv,OneModConj}
\end{document}

%% file: abstract.tex
Recently de Palma \textit{et al.} [\textit{IEEE Trans.~Inf.~Theory} 63, 728 (2017)] proved---using Lagrange multiplier techniques---that under a non-zero input entropy constraint, a thermal state input minimizes the output entropy of a pure-loss bosonic channel. In this note, we present our attempt to generalize this result to all single-mode gauge-covariant Gaussian channels by using similar techniques. Unlike the case of the pure-loss channel, we cannot prove that the thermal input state is the only local extremum of the optimization problem. 
However, we do prove that, if the conjecture holds for gauge-covariant Gaussian channels, it would also hold for gauge-contravariant Gaussian channels. The truth of the latter leads to a solution of the triple trade-off and broadcast capacities of quantum-limited amplifier channels. We note that de Palma \textit{et al.} [\textit{Phys.~Rev.~Lett.} 118, 160503 (2017)] have now proven the conjecture for all single-mode gauge-covariant Gaussian channels by employing a different approach from what we outline here. Proving a multi-mode generalization of de Palma {\em et al.}'s above mentioned result---i.e., given a lower bound on the von Neumann entropy of the input to an $n$-mode lossy thermal-noise bosonic channel, an $n$-mode product thermal state input minimizes the output entropy---will establish an important special case of the conjectured Entropy Photon-number Inequality (EPnI). The EPnI, if proven true, would take on a role analogous to Shannon's EPI in proving coding theorem converses involving quantum limits of classical communication over bosonic channels.


%% file: intro.tex
Finding the input to a channel---a probability distribution function for the case of classical channels and a density operator for the case of quantum channels---that minimizes the entropy of the channel output, is central to characterizing the fundamental limits of the channel's ability to transmit information reliably. For classical additive Gaussian noise channels, it is known that, subject to an input entropy lower bound, a Gaussian-distributed input minimizes the output entropy. An analogous statement has been conjectured for bosonic Gaussian channels (BGCs)---quantum channels that map Gaussian states of the input bosonic mode(s) to Gaussian states at the output. The conjecture states that, subject to a lower bound on the von Neumann entropy of the input state to a BGC, a thermal state input (one that has a circularly-symmetric complex Gaussian distribution in phase space) minimizes the entropy at the output of the channel. This conjecture is intimately tied to quantifying the communication capacity of bosonic channels, as discussed in Section~\ref{sec:background}. Various special cases of this conjecture, bounds, and related theorems have been proved in recent years~\cite{giovannetti2014ultimate,mari2014quantum,giovannetti2010generalized,giovannetti2015solution,de2015passive,de2016gaussian}. Recently, the minimum output entropy of non-Gaussian quantum channels have also been considered~\cite{memarzadeh2016nonGaussian}.

In this note, we present our attempt at proving the following conjecture:

\begin{conjecture}\label{thm:mainresult}
Consider a gauge-covariant bosonic Gaussian channel ${\cal N}_{A \to B}$, a trace-preserving completely-positive (TPCP) map that maps a quantum state $\rho_A$ of a single bosonic mode at the input of the channel to the single-mode state $\rho_B = {\cal N}_{A \to B}(\rho_A)$ at its output. For all input states $\rho_A$ with $S(\rho_A) \equiv S_0 > 0$, the output entropy $S(\rho_B)$ is minimized if $\rho_A$ is a thermal state of mean photon number $g^{-1}(S_0)$, where the entropy of a thermal state $\rho^{\rm th}_{\bar n}$ of mean photon number $\bar n$ is given by, $S(\rho^{\rm th}_{\bar n}) \equiv g({\bar n}) = ({\bar n}+1)\ln({\bar n}+1)-{\bar n}\ln {\bar n}$.
\end{conjecture}

Since the original arXiv posting of this note,  de Palma \textit{et al.}~have now proven that Conjecture~\ref{thm:mainresult} is true \cite{PTG17}, by employing a completely different approach from what we discuss here.
	
Conjecture~\ref{thm:mainresult}  is a generalization of the result proven in~\cite{giovannetti2015solution}, which provided a proof of Conjecture~\ref{thm:mainresult} when no lower-bound constraint was imposed on the input entropy, i.e., $S_0 = 0$. As we discuss next in Section~\ref{sec:background}, the $S_0 > 0$ case of this conjecture is required for applications to coding-theorem converse proofs for various problems involving communication over bosonic channels. However, the multi-mode version of Conjecture~\ref{thm:mainresult} is required in some cases, as discussed in Section~\ref{sec:background}, whose proof remains open.

It is well known that the thermal state is a local minimum of the optimization problem
set out in Conjecture~\ref{thm:mainresult}. However, the real difficult part of this sort of proof is to show that the thermal state is in fact a global minimum.  In Ref.~\cite{de2016gaussian}, de Palma \textit{et al.}~proved that the thermal state is the only critical point of the optimization function
for single-mode pure-loss BGCs, thus establishing Conjecture~\ref{thm:mainresult} for single-mode pure-loss BGCs. However, when we try to generalize their approach to general gauge-covariant BGCs (when there is an amplification component), we find that there is another critical point: a particular Fock state,  other than the thermal state solution. It is unclear to us whether this Fock state is indeed another local extremum with higher output entropy than that of a thermal state or if it can be eliminated as an invalid solution by imposing some physical conditions that we are unable to identify.

This paper is organized as follows. We begin with a discussion in Section~\ref{sec:background} on the relevance of the conjecture in the context of quantifying the quantum-limited capacity of single and multi-user optical communication channels and related results. In Section~\ref{sec:setup}, we give the definition of a gauge-covariant BGC and their Lindblad operator representation from \cite{giovannetti2010generalized}. In Section~\ref{sec:Th1}, we begin our proof attempt by reducing the aforesaid problem to an infinitesimal version thereof. In Section \ref{sec:th2}, we try to prove this infinitesimal version by solving a Lagrangian multiplier problem. We show that, unlike the case of pure-loss channels, there exists a solution to the Lagrangian equation other than the thermal state. In Section~\ref{sec:contra}, we prove that if the conjecture is true for gauge-covariant BGCs, it can be extended to gauge-contravariant BGCs. We conclude the paper in Section~\ref{sec:conclusions} with a summary and some open problems.

%% file: background.tex
The entropy power inequality (EPI) for statistically independent continuous-valued random variables $X$ and $Y$,
\begin{equation}
e^{2h(X+Y)} \ge e^{2h(X)} + e^{2h(Y)},
\end{equation}
was proposed by Shannon \cite{bell1948shannon}, where $h(X)$ is the differential entropy of $X$. Equality holds when $X$ and $Y$ are both Gaussian distributed. The EPI has found many applications in capacity theorem converse proofs for additive Gaussian noise channels in both single and multi-user scenarios. In all those applications, it suffices to restrict $Y$ to be Gaussian distributed. With $Y$ Gaussian and $\sigma_Y^2$ its variance, the EPI can be restated as
\begin{equation}
h(X+Y) \ge \frac{1}{2}\log\left(2\pi e \left[\sigma_Y^2 + \sigma_{X_G}^2\right]\right),
\label{eq:MOE_classical}
\end{equation}
where $\sigma_{X_G}^2 = \frac{e^{2h(X)}}{2\pi e} \equiv v(X)$ is the `entropy power' of $X$, i.e., the entropy of a Gaussian random variable $X_G$ whose entropy is the same as the entropy of $X$. One interpretation of~\eqref{eq:MOE_classical} is that, subject to a lower bound on the entropy of the input, i.e., $h(X) \ge H_0$, the output entropy of an additive white Gaussian noise (AWGN) channel is minimized when $X$ has a Gaussian distribution. It is this special-case interpretation of the EPI---that subject to an input entropy constraint, Gaussian inputs minimize the output entropy---that is crucial to all the known applications of the EPI to the coding theorem converse proofs. Intuitively, choosing an input distribution that minimizes the entropy at the output of a channel is the input that works with the channel's intrinsic noise in the most favorable way, with respect to the achievable and reliable communication rate.

Fundamental limits of optical communications are governed by quantum information theory (QIT), as optical channels are bosonic quantum channels at the core. Choices of transmitter modulation and receiver measurement induce a specific classical channel whose capacity is governed by its Shannon capacity. QIT gives us tools to evaluate the ultimate capacity of quantum channels without making restrictive assumptions on modulation states or structural assumptions on the receiver. 

Bosonic Gaussian channels (BGCs) are quantum channels that map Gaussian states at the input to Gaussian states at the output. Most practical scenarios in optical communications in fiber and free-space where the predominant effects are pure (linear) loss, thermal noise (resulting from the blackbody background for instance), and linear amplification, are modeled accurately by an appropriate BGC. Non-linear effects in optical propagation such as self Kerr, non-linear dispersion, and multi-photon absorption cannot be modeled by a BGC. In this paper, we will restrict our attention to phase-insensitive BGCs---both gauge-covariant and gauge-contravariant kinds (to be defined later in the paper), which subsume pure loss, additive thermal noise, linear amplification, and all combinations thereof.

We now recall some history concerning Conjecture~\ref{thm:mainresult}. The entropy photon-number inequality (EPnI) was proposed as a conjecture in 2008~\cite{guhaEPnI2008,guha2008entropy}.
Provided that it is true, it would 
 on a role exactly analogous to Shannon's EPI in proving coding theorem converses for quantifying the ultimate capacity of transmitting information over BGCs. The single-mode ($n=1$) version of the EPnI is stated in terms of the entropy powers of a pair of input states $\rho_X$ and $\rho_Y$ mixing on a beamsplitter of transmissivity $\eta$, producing an output state $\rho_Z$ (see Fig.~\ref{fig:1}), where the quantum generalization of entropy power (or, entropy photon-number) of $\rho_X$ is the power (mean photon number) of a thermal state with (von Neumann) entropy equal to $S(\rho_X) = -{\rm Tr}(\rho_X \ln \rho_X)$, the entropy of $\rho_X$. The entropy of a thermal state $\rho^{\rm th}_{\bar n}$ of mean photon number $\bar n$ is given by
\begin{equation}
S(\rho^{\rm th}_{\bar n}) = g({\bar n}) = ({\bar n}+1)\ln({\bar n}+1)-{\bar n}\ln {\bar n}.
\end{equation}
Therefore, the entropy photon number of $\rho_X$ is $g^{-1}(S(\rho_X))$. The statement of the EPnI is
\begin{equation}
g^{-1}(S_Z) \ge \eta\, g^{-1}(S_X) + (1-\eta)\,g^{-1}(S_Y),
\end{equation}
where $S_A = S(\rho_A)$; $A = X, Y$ or $Z$. The multi-mode ($n>1$) version of this conjecture is stated similarly~\cite{guhaEPnI2008}, but with entropy photon number of an $n$-mode input state $\rho_{X^n}$ defined as the mean photon number of an $n$-mode i.i.d. tensor-product thermal state with entropy $S(\rho_{X^n})$, which is given by $g^{-1}(S(\rho_{X^n})/n)$.

The restriction of the EPnI where one of the two inputs is held to be a thermal state (tensor-product thermal state for the multi-mode case) was stated in 2008 in~\cite{guha2008multiple} as Conjecture~3. The multi-mode ($n > 1$) version of Conjecture~3 suffices to prove all the coding theorem converses where the EPnI has been employed thus far. Specifically, a proof of Conjecture~3 would establish the quantum-limited classical capacity region of a single-sender multiple-receiver bosonic broadcast channel with loss and additive thermal-noise~\cite{guha2008multiple}.

A special case of Conjecture~3 of~\cite{guha2008multiple}, corresponding to $S_X = 0$ and ${\bar n} > 0$, stated as Conjecture 1 in~\cite{guha2008multiple} was conjectured in 2004~\cite{giovannetti2004minimum}, a proof of which was shown to establish that the capacity of the single-sender single-receiver lossy BGC with additive thermal-noise is attained by coherent-state inputs~\cite{giovannetti2004minimum}. Conjecture 1 states that the vacuum input state minimizes the output entropy of a lossy thermal-noise BGC. The proof of this conjecture for all multi-mode phase-insensitive BGCs appeared in 2014~\cite{mari2014quantum,giovannetti2014ultimate}, which used the fact that a quantum-limited conjugate amplifier channel is entanglement-breaking, and employed certain decomposition rules of a bosonic Gaussian channel. 

Another special case of Conjecture~3 of~\cite{guha2008multiple}, corresponding to $S_X > 0$ and ${\bar n} = 0$, stated as Conjecture 2 in~\cite{guha2008multiple}, was stated in 2007~\cite{guha2007classical}, a proof of which was shown to establish that the capacity of the single-sender multiple-receiver pure-loss BGC is attained by coherent-state inputs~\cite{guha2007classical}, and would complete the proof of the triple tradeoff region of a pure-loss BGC~\cite{2012_wilde_bosonictripletradeoff}. Conjecture 2 of~\cite{guha2008multiple} states that given an entropy constraint on the input of a lossy channel with zero added thermal noise, the output entropy is minimized by a thermal input state. A major step towards the proof of Conjecture 2 of~\cite{guha2008multiple} was made in 2015 in~\cite{de2015passive}, where it was shown that for all single-mode BGCs, a Fock-passive state (defined later in the paper) minimizes the output entropy among the set of all input states of a given eigen-spectrum. This result reduced an optimization over quantum states to an optimization of a (classical) probability mass function. More recently, Conjecture 2 of~\cite{guha2008multiple} was proven for the single-mode ($n=1$) case by considering the aforesaid minimization problem within the class of passive states~\cite{de2016gaussian}. The general proof of the multi-mode ($n>1$) version of Conjecture 2 remains open.

We again clarify that de Palma \textit{et al.}~have now proven that Conjecture~\ref{thm:mainresult} is true \cite{PTG17}, by employing a completely different approach from what we discuss here.
Also, it was shown in \cite{QW16} that Conjecture~\ref{thm:mainresult} implies the solution of the trade-off and broadcast capacities of quantum-limited amplifier channels, and so these latter problems are solved now in light of \cite{PTG17}. 


\begin{figure}[!t]
\centering
\includegraphics[width=\linewidth]{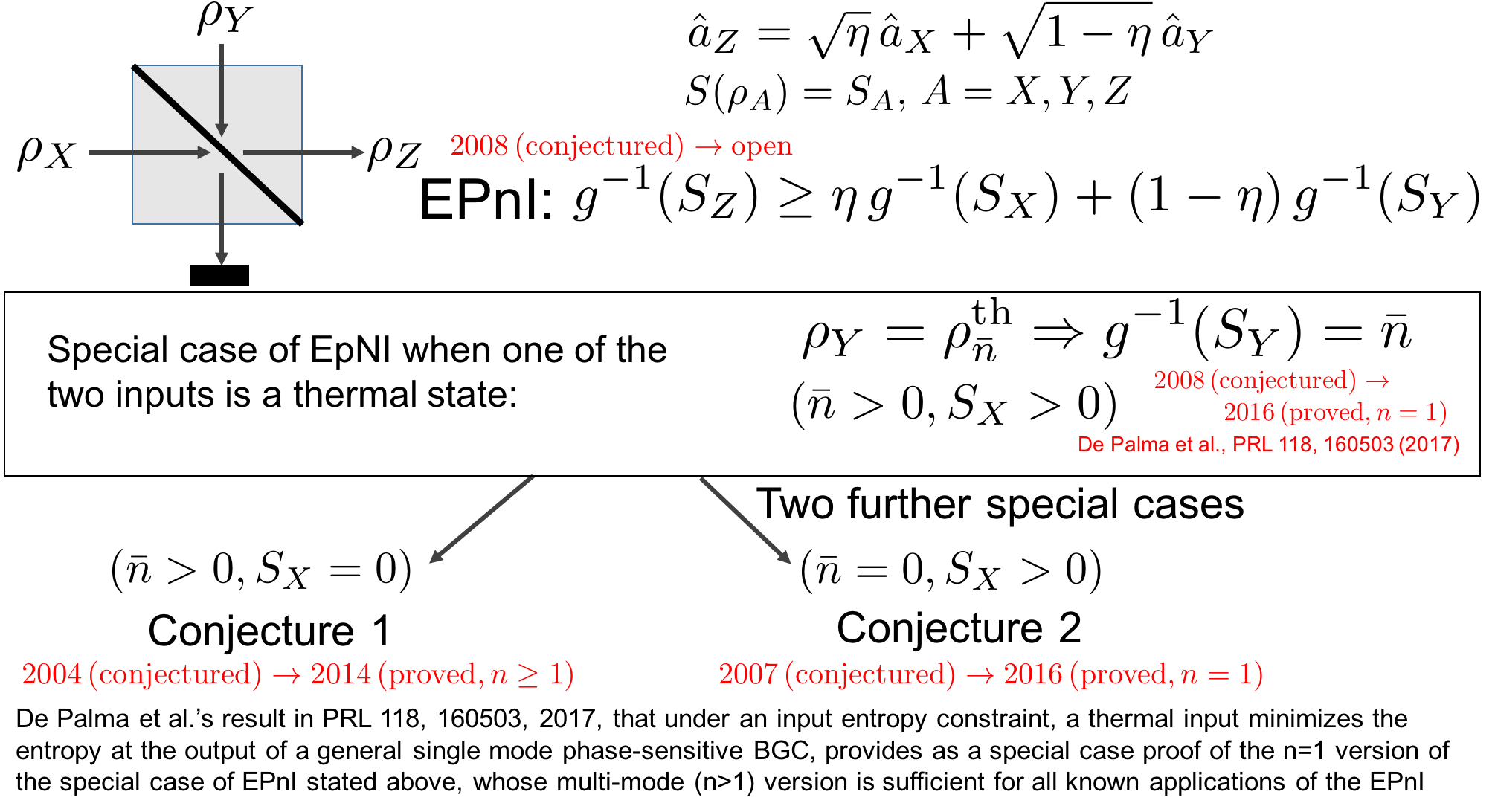}	
\caption{Entropy photon-number inequality and its special cases.}
\label{fig:1}
\end{figure}


Fig.~\ref{fig:1} summarizes the relationship between the different conjectures, the cases for which they have been proven and those that remain open.
Proving the multi-mode multi-input version of the EPnI would represent our ultimate understanding of the problem of minimum output entropy for bosonic channels.

%% file: mainresults.tex
Let us consider the Hilbert space of one mode of the electromagnetic field, a quantum harmonic oscillator. The annihilation operator $\hat{a}$ of a bosonic mode satisfies the following canonical commutation relation:
	\begin{align}
	[\hat{a},\hat{a}^\dagger] = 1,
	\end{align}
	from which we can build the Fock basis for the Hilbert space,
	\begin{align}
	\ket{n} = \frac{\hat{a}^{\dagger n}}{\sqrt{n!}}\ket{0},~~~\bra{m}n\rangle = \delta_{m,n}~.
	\end{align}
The Fock state $|n\rangle$, also known as the photon-number state, is a quantum state of a given mode corresponding to exactly $n$ photons being present in the mode. The annihilation operator and its conjugate (the creation operator) act on the Fock basis as
	\begin{align}
	\hat{a}\ket{n}&=\sqrt{n}\,\ket{n-1},\,{\text{and}}\\
	\hat{a}^\dagger\ket{n}&=\sqrt{n+1}\,\ket{n+1}.
	\end{align}
The general quantum state of a bosonic mode is a density operator $\rho$, an infinite-dimensional unit-trace positive semi-definite Hermitian operator. The Fock basis representation of $\rho$ (i.e., the infinitely-many matrix elements $\langle m | \rho |n \rangle$; $m,n = 0, 1, \ldots$) is equivalent to the characteristic function representation of $\rho$, $\chi_\rho(\xi) = {\rm Tr}(\rho\,e^{\xi {\hat a}^\dagger - \xi^* {\hat a}})$. A bosonic Gaussian channel (BGC) can be characterized by its action on the characteristic function of the input state. In particular, for one-mode gauge-covariant BGCs, the transformation of the input state $\rho_{\rm in}$ to the output state $\rho_{\rm out}$, expressed in terms of their characteristic functions takes the following form:
	\begin{align}
	\label{eq:characteristic}
	\chi_{\rho_{\rm out}}(\xi) = \chi_{\rho_{\rm in}}(\sqrt{\tau}\,\xi)\exp[-y|\xi|^2/2],
	\end{align}
	where $\tau >0$ is the loss or gain parameter and $y$ parametrizes added (Gaussian) noise. The BGC is a valid physical TPCP map if $y\geq|\tau-1|$. For a lossy bosonic channel of transmissivity $\eta \in (0, 1)$ and added thermal noise of mean photon number $N$, $\tau = \eta$ and $y = (1-\eta)(2N+1)$. For a phase-insensitive bosonic amplifier of gain $\kappa > 1$ and added thermal noise of mean photon number $N$, $\tau = \kappa$ and $y = (\kappa - 1)(2N+1)$. For a unit-gain additive-noise channel with photon-number-unit noise variance $N$, $\tau = 1$ and $y = 2N$. It was shown in~\cite{giovannetti2010generalized}, starting from the characteristic-function description, that all single-mode gauge-covariant BGCs possess a semi-group structure, and consequently that a gauge-covariant single-mode BGC can be represented as a one-parameter linear TPCP map,
	\begin{align}
	\rho(t)=\Phi_t(\rho) = e^{t\L}(\rho)~,
	\end{align}
	with manifest semi-group structure 
	\begin{align}
	e^{(t+t')\L} = e^{t'\L}e^{t\L}=e^{t\L}e^{t'\L}~.
	\end{align}
	Here $\L$ is a Lindblad operator that generates the dynamics of a gauge-covariant BGC, and $t$ can be viewed as a time parameter corresponding to a continuous action of the channel on the input state $\rho(0)$, resulting in the final state $\rho(t)$ at time $t > 0$. With that interpretation, the equation of motion for $\rho(t)$, under the action of the channel, is given by
	\begin{align}
	\frac{d\rho(t)}{dt} = \L(\rho(t))~.
	\end{align}
	For gauge-covariant BGCs, the Lindblad operator is given by
	\begin{align}
	\L = \gamma_+ \L_+ + \gamma_-\L_-~,
	\end{align}
where
	\begin{align}
	\L_+(\rho) &= {\hat a}^\dagger\rho {\hat a} -\frac{1}{2}{\hat a}{\hat a}^\dagger\rho -\frac{1}{2} \rho {\hat a}{\hat a}^\dagger ~,\\
	\L_-(\rho) &= {\hat a}\rho {\hat a}^\dagger -\frac{1}{2}{\hat a}^\dagger {\hat a} \rho -\frac{1}{2} \rho {\hat a}^\dagger {\hat a}~.
	\end{align}
	For a lossy channel with thermal noise, $\gamma_+ = N,\gamma_- = N+1$, where $N$ is the mean photon number of the thermal state. The channel transmissivity $\eta = e^{-t}$. For a phase-insensitive noisy amplifier channel, $\gamma_+ = N+1,\gamma_- = N$, and the amplifier gain $\kappa = e^{t}$. Finally, for an additive Gaussian noise channel, $\gamma_+=\gamma_- = 1$ with $N=t$.
	
	The density operator of a thermal state with mean photon number $\bar n$ is diagonal in the Fock basis and  is given by 
	\begin{align}
	\rho^{\rm th}_{\bar n} &= \frac{1}{1+{\bar n}}\left(\frac{\bar n}{1+{\bar n}}\right)^{\hat{a}^\dagger\hat{a}}~,\\
	&= \sum_{n=0}^{\infty} \lambda_n |n\rangle \langle n|,
	\end{align}
where $\lambda_n = {\bar n}^n/(1+{\bar n})^{n+1}$ is the geometric (Bose-Einstein) distribution. The von Neumann entropy of the thermal state is  
	\begin{align}
	S(\rho^{\rm th}_{\bar n}) = g({\bar n}) = ({\bar n}+1)\ln({\bar n}+1)-{\bar n}\ln {\bar n}~,
	\end{align}
	which clearly is also the (discrete) Shannon entropy of the geometric distribution $\left\{\lambda_n\right\}$.
	
	We now restate and specialize 
	Conjecture~\ref{thm:mainresult}, in terms of the aforementioned description of a gauge-covariant BGC (see Section~\ref{sec:contra} for our statement for gauge-contravariant channels).
	\begin{conjecture}\label{th:finite-conjecture}
	Consider a single-mode gauge-covariant BGC represented by $\Phi_t = e^{t\L}$. For any input state $\rho$ with a given entropy $S(\rho) = S_0 \geq 0$, the output entropy is minimized by the input thermal state, i.e.,
	\begin{align}\label{eq:finite-conjecture}
	S(\Phi_t(\rho))\geq S(\Phi_t(\rho^{\rm th}_{g^{-1}(S_0)}))~,~~\forall t\geq 0~.
	\end{align}
	\end{conjecture}
	
The goal of the remainder of this note is to discuss our attempt at a proof of Conjecture~\ref{th:finite-conjecture}. For the case when the input state $\rho$ is a pure state, i.e., $S_0=0$, the statement of Conjecture~\ref{th:finite-conjecture} was proved in~\cite{giovannetti2014ultimate}. We will try to prove it for the general case, i.e., $S_0>0$. In Section \ref{sec:Th1}, we will reduce the statement of Conjecture~\ref{th:finite-conjecture} to Conjecture~\ref{th:infinite-conjecture}, an equivalent infinitesimal version of Conjecture~\ref{th:finite-conjecture}, but stated only at $t=0$. Finally, in Section~\ref{sec:th2}, we will discuss our proof attempt for Conjecture~\ref{th:infinite-conjecture} using a Lagrange-multiplier optimization and show where we are unable to complete the proof.

%% file: reduction.tex
	In this section, we will show that proving the following infinitesimal version of the minimum output entropy conjecture established in \cite{giovannetti2010generalized}, which we state entirely at $t=0$, implies a proof of Conjecture \ref{th:finite-conjecture} for all $t \ge 0$.
	\begin{conjecture}\label{th:infinite-conjecture}
	Consider a single-mode gauge-covariant BCG represented by $\Phi_t = e^{t\L}$. For any input state $\rho$ with entropy $S(\rho) = S_0 > 0$, we have 
	\begin{align}
	\label{eq:inf_conjecture}
	\frac{d}{dt}S(\Phi_t(\rho))|_{t=0}\geq \left. \frac{d}{dt}S\!\left(\Phi_t(\rho^{th}_{g^{-1}(S_0)})\right)\right|_{t=0}~.
	\end{align}
\end{conjecture}

See Section \ref{sec:th2} for a discussion of our attempt at proving Conjecture~\ref{th:infinite-conjecture}.


\begin{lemma}
Conjecture~\ref{th:infinite-conjecture} implies Conjecture \ref{th:finite-conjecture}.
\end{lemma}

\begin{proof}
Let us assume that the statement of Conjecture~\ref{th:infinite-conjecture}, Eq.~\eqref{eq:inf_conjecture} is true. Due to the semi-group property of the map $\Phi_t = e^{t\L}$, we can discretize it as a concatenation of $N$ evolutions:
\begin{align}
e^{t\L} = \prod_{j=1}^{N}e^{j (\delta t) \L}~,
\end{align} 
where $\delta t = t/N$.

$\Phi_0(\rho) = \rho$ is the identity map. Therefore, the input entropy constraint $S(\rho) = S_0$ can be restated as
\begin{align}\label{eq:initial}
S(\Phi_0(\rho)) = S(\Phi_0(\rho^{\rm th}_{g^{-1}(S_0)}))~.
\end{align}
Now consider the output entropy at $t = \delta t$.  From Conjecture \ref{th:infinite-conjecture} and the initial condition (\ref{eq:initial}), for $\delta t$ infinitesimally small, we must have that the thermal state minimizes the output entropy, 
\begin{align}
S(\Phi_{\delta t}(\rho))\geq S(\Phi_{\delta t}(\rho^{\rm th}_{g^{-1}(S_0)}))\equiv S_{\delta t}~.
\end{align}
Now to minimize the entropy of the output state at $t'=2(\delta t)$, we need to consider the following minimization problem:
\begin{align}
\min_{S(\rho(\delta t))\geq S_{\delta t}}S(\Phi_{\delta t}(\rho(\delta t)))~.
\label{eq:minatdeltat}
\end{align}
Due to the concavity of von Neumann entropy and the linearity of $\Phi_t$, we can restrict the minimization in~\eqref{eq:minatdeltat} to a smaller set $S(\rho(\delta t))= S_{\delta t}$ without affecting the minimum~\cite{giovannetti2010generalized},
\begin{align}
\min_{S(\rho(\delta t))\geq S_{\delta t}}S(\Phi_{\delta t}(\rho(\delta t))) = 
\min_{S(\rho(\delta t))= S_{\delta t}}S(\Phi_{\delta t}(\rho(\delta t)))~.
\end{align}
Using Conjecture~\ref{th:infinite-conjecture} again for the initial state at $t=\delta t$ with the modified input-entropy constraint
\begin{align}
S(\rho(\delta t))= S_{\delta t}~,
\end{align}
we obtain
\begin{align}
\frac{d}{dt'}S(\Phi_{t^\prime}(\rho(\delta t)))|_{t'=0}\geq \frac{d}{dt'}S\!\left(\Phi_{t^\prime}((\rho^{th}_{g^{-1}(S_{\delta t})}))\right)|_{t'=0}~.
\end{align}
Therefore, for an infinitesimally small $\delta t$, we must have that a thermal state minimizes the output entropy at $t=2(\delta t)$,
\begin{align}
S(\Phi_{2(\delta t)}(\rho))\geq S(\Phi_{\delta_t}(\rho^{th}_{g^{-1}(S_{\delta t})})) = S(\Phi_{2(\delta t)}(\rho^{th}_{g^{-1}(S_0)})~.
\end{align}
We repeat the above procedure for $t = 3(\delta t), \ldots , N(\delta t)$, and then take the limit $N\rightarrow \infty$ (or equivalently $\delta t \to 0$), thereby concluding that a thermal state of entropy $S_0$ minimizes the entropy at the output of $\Phi_t$ at any arbitrary finite $t > 0$, which is the statement of Conjecture~\ref{th:finite-conjecture} in Eq.~\eqref{eq:finite-conjecture}. 
\end{proof}

%% file: mainproof.tex
Given the reduction in Section~\ref{sec:Th1}, what  remains to be proven in order to establish Conjecture~\ref{thm:mainresult} for gauge-covariant channels is Conjecture~\ref{th:infinite-conjecture}. Here we discuss our attempt at proving 
Conjecture~\ref{th:infinite-conjecture}.

We first argue that it suffices to restrict the minimization in Conjecture~\ref{th:infinite-conjecture} to {\em passive} input states. Passive states are diagonal in the Fock basis and have non-increasing eigenvalues:
\begin{equation}
\mathtt{p}= \sum_{n=0}^\infty p_n\ket{n}\bra{n},\qquad p_0\geq p_1 \geq \ldots\geq 0~.
\end{equation}	
Since the probabilities $p_n$ are decreasing, we can conclude that all passive states have connected supports~\cite{de2016gaussian}. That is, there exists $N\geq 0$ such that $p_0\geq p_1 \geq\ldots\geq p_N >p_{N+1}=\ldots=0$. 
We define $\D_N$ as the set of all such passive states, and we define the set
$\D_\infty$ as the set of passive states with infinite support. Therefore, all the passive states we need to consider in this problem belong to the following set:
\begin{align}
\D = \bigcup_{N=1}^\infty \D_N~.
\end{align}
The Fock passive rearrangement $\rho^{\downarrow}$ of the state $\rho$ is the only passive state with the same spectrum as that of $\rho$. 

Ref.~\cite{de2015passive} established that {\em passive} states minimize the output entropy of single-mode phase-insensitive Gaussian channels.

\begin{theorem}\label{th:reduce-passive}
For single-mode gauge-covariant BGCs, 
the minimum output entropy is achieved on the set of  passive input states:
\begin{align}
\min_{S(\rho)=S_0}S(\Phi_t(\rho)) = \min_{S(\p) = S_0, \p\in \D}S(\Phi_t(\p)).
\end{align} 
\end{theorem}
\begin{proof}
	Theorem V.3 in \cite{de2015passive} implies that for any input state $\rho$,
	\begin{align}
	S(\Phi_t(\rho))\geq S(\Phi_t(\rho^{\downarrow}))~.
	\end{align}
	Thus the claim follows.
\end{proof}

It therefore suffices to prove that Conjecture \ref{th:infinite-conjecture} holds on the set of passive input states.
Now we will calculate the derivative of the output entropy when the input is a passive state. We will first argue that we must restrict the optimization over input passive states with support over all the (infinitely many) Fock basis elements. 

\begin{lemma}\label{lemma:infinite-support}
The passive input state that achieves the minimum output entropy of $\Phi_t = e^{t \L}$ cannot have a finite support on the Fock basis, i.e., zero probability over basis elements $\left\{|N+1\rangle, |N+2\rangle, \ldots\right\}$ for some finite $N$. 
\end{lemma}

\begin{proof}
Let $\p^{(N)}(0) \in \D_N$ be the input passive state, and let $\p^{(N)}(t) = \Phi_t(\p^{(N)}(0)) \equiv \sum_{n=0}^{\infty}p_n(t)|n\rangle \langle n|${\footnote{Note that $\p^{(N)}(t)$ may have support over Fock basis elements higher than $|N\rangle$, as reflected by the upper limit of the sum.}}. By assumption, we have that $p_n(0) = 0$ for all $n \ge N+1$.
	\begin{align}\nonumber
	&\frac{d}{dt}(e^{t\L} \p^{(N)}(t))|_{t=0}
	=\L (\p^{(N)}(0))\\
	\nonumber
	 & = \sum_{n=0}^N p_n(0) \gamma_+ (a^\dagger\ket{n}\bra{n} a -\frac{1}{2}aa^\dagger\ket{n}\bra{n} -\frac{1}{2} \ket{n}\bra{n} aa^\dagger)\\
	 \nonumber
	&\quad+\sum_{n=0}^N p_n(0) \gamma_-(a\ket{n}\bra{n} a^\dagger -\frac{1}{2}a^\dagger a\ket{n}\bra{n} -\frac{1}{2} \ket{n}\bra{n} a^\dagger a)\\
	\nonumber
	&=\sum_{n=0}^N p_n(0) \gamma_+ ((n+1)\ket{n+1}\bra{n+1} -(n+1)\ket{n}\bra{n})\\
	&\quad+\sum_{n=0}^N p_n(0) \gamma_- (n\ket{n-1}\bra{n-1} -n\ket{n}\bra{n})~.
	\end{align}
	where, in the above, we take the convention that
	$\left\vert -1\right\rangle \!\left\langle -1\right\vert$ is the zero operator.
	Therefore for $n = 0, 1,\ldots, N$, we have the following equations:
	\begin{multline}\label{eq:p'0}
	p_n'(0)
	=\gamma_+( np_{n-1}-(n+1)p_n) \\+\gamma_-((n+1)p_{n+1}-np_n)~,
	\end{multline}
	where it is implicit that $p_{-1}(0) = 0$ and $p_{N+1}(0) = 0$.
	By Taylor expanding the operator
	$e^{t\L}$ as $\sum_{k=0}^\infty ({t\L})^k / k!$, one can see that
	when $\gamma_+>0$, for all $t >0$
	  the state $\p^{(N)}(t)\in \D_{\infty}$.
	For all initial states with finite support (i.e.,
	$\p^{(N)}(0)\in \D_N$), we can calculate the derivative of the output entropy as follows:
	\begin{align}
	&\left.\frac{d}{dt}S(e^{t\L}\p^{(N)}(0))\right|_{t=0}\nonumber\\
	& = \left.\frac{d}{dt}S(\p^{(N)}(t))\right|_{t=0}
	\nonumber\\
	& = \frac{d}{dt}\left.\left[
	-\sum_{n=0}^{\infty} p_n(t) \ln p_n(t)\right]\right|_{t=0}\nonumber\\
	& = \left.\left[
	-\sum_{n=0}^{\infty} (1+\ln p_n(t))p_n'(t)\right]\right|_{t=0}	
	\end{align}
	To evaluate this latter limit as $t\to 0$,
	consider from \eqref{eq:p'0} that all terms
	$p_n'(0)$ with $n>N+1$ converge to zero, whereas
	all terms $p_n(0)$ with $n > N$ converge to zero.
	We can then write out the last line above explicitly
	as
	\begin{multline}
	\frac{d}{dt}S(\p^{(N)})|_{t=0} 
	= \\-\sum_{n=0}^{N} (1+\ln p_n(0))\Big[\gamma_+ np_{n-1}(0)-\gamma_+(n+1)p_n(0) \\
	 -\gamma_-n p_{n}(0) +\gamma_- (n+1)p_{n+1}(0)\Big]\\
	  - \gamma_+ (N+1) \lim_{t \to 0}
	  p_N(t)(1+\ln p_{N+1}(t))~.
	\end{multline}
	For $\gamma_+ >0$, 
	\begin{align}
	 - \gamma_+ (N+1) \lim_{t \to 0}
	  p_N(t)(1+\ln p_{N+1}(t))\rightarrow +\infty~.
	\end{align} 
	Thus, the derivative of the output entropy is infinitely large for any initial passive state with finite support. That is to say, any passive state $\p^{(N)}\in \D_N$ with $N<+\infty$ cannot be a local minimum of the derivative of output entropy. Therefore, we only need to search for a minimum within the set of passive states having infinite support, that is, $\p^{(\infty)}\in \D_\infty$.\footnote{Note that when $\gamma_+=0$, the channel is a pure-loss channel, and the derivative does not diverge in this case. However, following the proof in \cite{de2016gaussian}, it can be shown as well that it is sufficient to consider passive states in $\D_\infty$.}
\end{proof}	
	
	For the last part of our discussion, we  restrict the minimization to be over passive input states with an infinite support over the Fock basis. This optimization can be formulated as the following Lagrange multiplier problem:
	\begin{multline}\label{eq:Lagragian}
	L = S'(\p^{(\infty)}(0))+ \mu\left(-\sum_{n=0}^\infty p_n(0)\ln p_n(0)-S_0\right)\\
	 + \lambda \left(\sum_{n=0}^\infty p_n(0)-1\right)~,
	\end{multline}
	where,
	\begin{multline}	\label{eq:derivative}
	S'(\p^{(\infty)}(0)) = \\
	-\sum_{n=0}^{\infty} (1+\ln p_n(0))\Big[\gamma_+ np_{n-1}(0)-\gamma_+(n+1)p_n(0)\\
	-\gamma_-n p_{n}(0) +\gamma_- (n+1)p_{n+1}(0)\Big] ~.
	\end{multline}
	Any local minimum or maximum must be a solution of 
	\begin{align}\label{eq:mini}
	\nabla L = \left(\frac{\partial L}{\partial p_n},\frac{\partial L}{\partial \mu}, \frac{\partial L}{\partial \lambda}\right) = 0~.
	\end{align}

	Differentiating (\ref{eq:Lagragian}) with respect to $p_n(0)$ and using (\ref{eq:derivative}), we have that
	\begin{multline}
	-\gamma_+ n z_{n-1}^{-1}+\gamma_+ (n+1) + \gamma_-n -\gamma_- (n+1)z_n\\
	 +\gamma_+ (n+1)\ln z_{n}^{-1} + \gamma_-n\ln z_{n-1}\\-\mu(1+\ln p_n(0))+\lambda = 0~,
	\end{multline}
	where in the above we define 
	\begin{align}
	z_n = \frac{p_{n+1}(0)}{p_n(0)}~.
	\end{align}
	Since $\mathtt{p}_\infty \in \D_\infty$, $0<z_n\leq 1$.
	To get a recursive relation between $\brac{z_n}$, take the difference between two consecutive values of $n$, we get
	\begin{multline}
	\gamma_+n\left(\frac{1}{z_n}-\frac{1}{z_{n-1}}\right)+\gamma_-(n+2)(z_{n+1}-z_n)\\
	=\gamma_++\gamma_- - \gamma_+ \frac{1}{z_n}-\gamma_-z_n + \gamma_+(n+2)\ln\frac{1/z_{n+1}}{1/z_{n}}\\
	+\gamma_- n\ln\frac{z_n}{z_{n-1}}+(-\gamma_++\gamma_- -\mu)\ln z_n~.
	\end{multline}
	For the pure-loss channel, set $\gamma_+ =0, \gamma_- = 1$, we then recover (V.30) in Ref.\cite{de2016gaussian}. The above equation can be rewritten as following,
	\begin{align}\label{eq:master}
	(n+2)[f(z_{n+1})-f(z_{n})]=n[g(z_{n})-g(z_{n-1})]+\Delta(z_n)
	\end{align}
	for $n \geq 1$. When $n=0$, we have 
	\begin{align}\label{eq:master-0}
	2[f(z_{1})-f(z_{0})]=\Delta(z_0)~,
	\end{align}
	where
	\begin{align}
	\label{eq:fx}
	f(x) &= \gamma_-x+\gamma_+\ln x~,\\
	g(x) &= \gamma_-\ln x - \gamma_+ \frac{1}{x}~,
	\end{align}
	both of which are monotonic increasing with well-defined inverse function for $\gamma_+\geq 0$ and $\gamma_-\geq 0$. And we define $\Delta(x)$ as
	\begin{align}\label{eq:Delta}
	\Delta(x) = \gamma_++\gamma_- -\gamma_+\frac{1}{x}-\gamma_-x +(\gamma_--\gamma_+-\mu)\ln x~.
	\end{align}
	Let us also define the function
	\begin{align}
	h(x)=\frac{\gamma_-x+\gamma_+\frac{1}{x}-(\gamma_++\gamma_-)}{\ln x}~,
	\end{align}
	which is strictly increasing since both functions $(x-1)/\ln x$ and $(1/x-1)/\ln x$ are strictly increasing.
	Then we have the following lemma:
	\begin{lemma}\label{lemma:monotonic}
	We must have that $\brac{z_n}$ is either strictly increasing, strictly decreasing or is ¡a constant. In particular, if $h(z_0)> \gamma_--\gamma_+ -\mu$,  $\brac{z_n}$ is strictly increasing, $0<z_0\leq z_1<\ldots\leq1$. On the other hand, if $h(z_0)< \gamma_--\gamma_+ -\mu$, $\brac{z_n}$ is strictly decreasing, $1\geq z_0> z_1 >\ldots >0$. And $\brac{z_n}$ is a constant if and only if $h(z_0)=\gamma_--\gamma_+ -\mu$.
	\end{lemma}
	\begin{proof}
		Let us first consider $h(z_0)> \gamma_--\gamma_+ -\mu$. We want to prove $0<z_0< z_1<\ldots\leq1$. From the definition Eq.~(\ref{eq:Delta}) we have $\Delta(z_0)> 0$. Then from Eq.~(\ref{eq:master-0}), we have
		\begin{align}
		2[f(z_1)-f(z_0)]>0~.
		\end{align} 
		Since $f(x)$ is strictly increasing, we must have
		\begin{align}
		z_1>z_0 >0~.
		\end{align}
		Now let us assume that $z_n>\ldots> z_1>z_0>0$ is true. We will prove that $z_{n+1}> z_n$. Since $h(x)$ is increasing and $z_n> z_{0}$, we must have
		\begin{align}
		h(z_n)> h(z_0)>\gamma_--\gamma_+ -\mu~,
		\end{align}
		which in turn gives us $\Delta(z_n)> 0$. Also because $g(x)$ is strictly increasing, we have $g(z_{n})-g(z_{n-1})> 0$. Then from the recursive relation Eq. (\ref{eq:master}), we have
		$(n+2)[f_{z_n+1}-f(z_n)] > 0$. Since $f(x)$ is strictly increasing, we have 
		\begin{align}
		z_{n+1}> z_n~.
		\end{align}
		By induction for $n=2,3,\ldots$, we have $0<z_0< z_1<\ldots\leq1$. 
		Therefore in this case $\brac{z_n}$ must be strictly increasing.
		
		When $h(z_0)< \gamma_--\gamma_+ -\mu$, we have $\Delta(z_0)< 0$, it's not difficult to follow a similar induction procedure to prove
		$1\geq z_0> z_1 > \ldots \geq 0$.
		
		Finally when $h(z_0)=\gamma_--\gamma_+ -\mu$, the right hand side of the recursive equation Eq. (\ref{eq:master}) is always zero, thus $\brac{z_n}$ is a constant.
	\end{proof}
	\begin{lemma}\label{lemma:z!=1}
		We have $\lim\limits_{n\rightarrow \infty}z_n\neq 1$.
	\begin{proof}
		From Lemma \ref{lemma:monotonic}, $\brac{z_n}$ is either monotonically increasing or decreasing. When $\brac{z_n}$ is monotonically decreasing and if $\lim\limits_{n\rightarrow \infty}z_n= 1$, we must have $z_0 = z_1=\ldots=1$. Thus we have a uniform distribution with entropy
		\begin{align}
		S = \lim_{n\rightarrow \infty}\log n,
		\end{align}
		which is infinitely large. This contradicts  the finite constraint $S(\rho) = S_0$. Therefore $\lim\limits_{n\rightarrow \infty}z_n\neq 1$.
		
		Now consider the case when $\brac{z_n}$ is monotonically increasing and assume $\lim\limits_{n\rightarrow \infty}z_n= 1$. Then for arbitrary $\epsilon >0$, there exists a large integer $N$ such that $|p_n-p_N|\leq \epsilon, ~\forall n>N$. The entropy of the state is 
		\begin{align}
		\nonumber
		S &= -\sum_{n=0}^N p_n(0)\log p_n(0) -\sum_{n=N+1}^\infty p_n(0)\log p_n(0)\\
		&=-\sum_{n=0}^N p_n(0)\log p_n(0) - P_N\log(1-P_N)\nonumber\\
		&\qquad  +\lim_{n\rightarrow \infty}P_N\log n +o(\epsilon)~,
		\end{align}
		where $P_N = \sum_{n=0}^N p_n(0)$. Again $S$ diverges, and so $\brac{z_n}$ cannot be a valid solution. Therefore for any solution of (\ref{eq:master}), $\lim\limits_{n\rightarrow \infty}z_n\neq 1$~.
	\end{proof}
	\end{lemma}
	
	\begin{lemma}
		\label{lemma:geometric}
		We must have $z_0=z_1=\ldots=z$, with $0< z<1$, or $z_0>z_1>\ldots>z$, with $z=0$.
	\end{lemma}
	\begin{proof}
		From Lemma \ref{lemma:monotonic}, the series $\brac{z_n}$ is either monotonic increasing or decreasing. Since $0\leq z_n \leq 1$, $z_n$ must converge at large $n$,
		\begin{align}
		\lim\limits_{n\rightarrow \infty}z_n = z~,
		\end{align} 
		where $0 \leq z \leq 1$. From Lemma \ref{lemma:z!=1} we know that $z\neq 1$. 
		When $\brac{z_n}$ is increasing, $z\neq 0$ since $z_n$ would be negative. When $\brac{z_n}$ is a constant,  $z\neq0$ represents the zero-temperature limit -- the vacuum state.
		
		Let us first consider the case when $h(z_0)> \gamma_--\gamma_+ -\mu$ and $\brac{z_n}$ is increasing.  
		We then have $g(z_{n})-g(z_{n-1}) > 0$. From Eq. (\ref{eq:master})  we get 
		\begin{align}
		\Delta(z_n)- (n+2)[f(z_{n+1})-f(z_{n})]< 0~.
		\end{align}
		 Since $\ln z_n <0$, we can divide it on both sides and have
		\begin{align}
		-h(z_n) +(\gamma_--\gamma_+-\mu)> \frac{(n+2)[f(z_{n+1})-f(z_{n})]}{\ln z_n}~.
		\end{align}
		Now take the limit of $n\rightarrow \infty$ on both sides and since
		\begin{align}
		\nonumber
		&\lim\limits_{n\rightarrow \infty}\frac{(n+2)[f(z_{n+1})-f(z_{n})]}{\ln z_n}\\
		\nonumber
		=&\lim\limits_{n\rightarrow \infty} (n+2)\frac{\gamma_-(z_{n+1}-z_n)+\gamma_+(\ln z_{n+1}-\ln z_{n})}{\ln z_n}\\
		=& \gamma_-\lim\limits_{n\rightarrow \infty} (n+2)\frac{z_{n+1}-z_n}{\ln z}~,
		\end{align}
		where we use the fact 
		\begin{align}\label{eq:ratio-limit}
		\lim_{n\rightarrow \infty}\frac{\ln z_{n+1}}{\ln z_n}=1
		\end{align}
		for $0<z<1$.
		We then get
		\begin{align}\label{eq:step1}
		-h(z_n) +(\gamma_--\gamma_+-\mu)>\gamma_-\lim\limits_{n\rightarrow \infty}\frac{n(z_{n+1}-z_n)}{\ln z}~.
		\end{align}
		Define $x_n = z_{n+1}-z_n$ and consider the summation
		\begin{align}\label{eq:sum-bound}
		\sum_{n=0}^{\infty} x_n= z-z_0 <+\infty~. 
		\end{align}
		Since the summation of $x_n$ converges, we must have that $x_n$ decreases with $n$ and $x_n \sim o(1/n)$. Thus
		\begin{align}
		\lim\limits_{n\rightarrow \infty} n(z_{n+1}-z_{n})=0~.
		\end{align}
		Together with (\ref{eq:step1}) we have
		\begin{align}
		h(z)< \gamma_--\gamma_+-\mu,
		\end{align}
		Recall that at the beginning we assumed that  $h(z_0)> \gamma_--\gamma_+ -\mu$. Then we have a contradiction
		\begin{align}
		\gamma_--\gamma_+-\mu < h(z_0)< h(z_n)<	h(z)< \gamma_--\gamma_+-\mu~.
		\end{align}
		Therefore we cannot have $\brac{z_n}$ being strictly increasing.
		
		When $h(z_0)> \gamma_--\gamma_+ -\mu$ and $\brac{z_n}$ is strictly decreasing we can follow a similar argument, except the relation
		\begin{align}
		\lim_{n\rightarrow \infty}\frac{\ln z_{n+1}}{\ln z_n}=1
		\end{align}
		does not necessarily hold when $z=0$. Therefore, any solution with strictly decreasing $\brac{z_n}$ and $z>0$ will lead to contradiction.
		Combining the above argument with Lemma \ref{lemma:monotonic} we can conclude the proof.
%
	\end{proof}

As we can see, other than the thermal state (geometric solution), another Fock passive state with super-exponential decreasing probability distribution is also a valid solution. Mathematically, the inability to eliminate this Fock-diagonal state as a valid solution lies at the fact that \eqref{eq:ratio-limit} does not hold for this particular state. Notice that when $\gamma_+ = 0$, for pure-loss channel, we do not have this problem at all.  That is saying, the key is that the proof is not working for quantum-limited amplifier. It is unclear for us what is the physical implication of this difference between the pure-loss channel and the quantum-limited amplifier channel.

%% file: contravariant.tex
Here we show that  Conjecture \ref{th:finite-conjecture} holding for all single-mode gauge-covariant channels implies that it holds for all single-mode gauge contravariant channels. The transformation of the characteristic function induced by a contravariant channel $\tilde{\A}_\kappa^N $ is given by
\begin{align}\label{eq:contra-charateristic}
\chi_{\rho_{\rm out}}(\xi) = \chi_{\rho_{\rm in}}(-\sqrt{|\tau|}\xi^*)\exp[-y|\xi|^2/2],
\end{align}
with $\tau=-(\kappa-1)$ and $y\geq |\tau-1|$. The channel is quantum-noise limited when $y = |\tau-1|$. Let us first consider the quantum-limited gauge-contravariant channel $\tilde{\A}_\kappa^0$, which is an entanglement-breaking channel and thus can be expressed as a measure-prepare mapping in the coherent-state basis\cite{giovannetti2014ultimate},
\begin{align}
\tilde{A}_{\kappa}^0 (\rho) = \int d^2\xi \frac{\bra{\xi}\rho\ket{\xi}}{\pi}\ket{-\sqrt{\kappa}\xi^*}\bra{-\sqrt{\kappa}\xi^*}~.
\end{align}
It is instructive to note here that $\frac{\bra{\xi}\rho\ket{\xi}}{\pi}$ is the probability distribution resulting from an ideal heterodyne detection measurement on $\rho$ with measurement result $\xi \in {\mathbb C}$. So, the action of the channel $\tilde{\A}_\kappa^0$ can be regarded as heterodyne detection of $\rho$, obtaining measurement result $\xi$, followed by a conditional preparation of the coherent state $|-\sqrt{\kappa} \xi^* \rangle$. 

It is known that there exists an entanglement-breaking channel $\M$ \cite{giovannetti2015solution}
\begin{align}
\M (\rho) = \int d^2\xi \frac{\bra{\xi}\rho\ket{\xi}}{\pi}\ket{-\sqrt{\kappa}\xi}\bra{-\sqrt{\kappa}\xi}~,
\end{align}
which is a gauge covariant channel with parameters $\tau = \kappa-1$ and $y=\kappa$ as stated in (\ref{eq:characteristic}), such that
\begin{align}
\tilde{\A}_\kappa^0(\rho) = T\M(\rho)~.
\end{align}
In the above, $T$ is the super operator corresponding to complex conjugation, i.e., $T\hat{a}=\hat{a}^\dagger$ \cite{giovannetti2015solution}.
$T$ is not a TPCP map and hence not a valid quantum channel, but it is trace-preserving. Therefore $\tilde{\A}_\kappa^0(\rho)$ and $\M(\rho)$ must have the same spectrum. So, there must exist a unitary $U$ such that,
\begin{align}
\tilde{\A}_\kappa^0(\rho) =U^\dagger\M(\rho)U,
\end{align}
and thus,
\begin{align}
\min_{S(\rho)=S_0}S(\A_\kappa^0(\rho)) = \min_{S(\rho)=S_0}S(\M(\rho))~.
\end{align}
Thus Conjecture \ref{th:finite-conjecture} holds  for the quantum-noise-limited gauge-contravariant channel
if it holds for single-mode, gauge-covariant bosonic Gaussian channels. 

What remains is to consider a non-quantum-noise limited gauge-contravariant BGC (with thermal noise). A contravariant channel with parameter $(\tau,y)$ as in (\ref{eq:contra-charateristic}) can be decomposed into a quantum-limited attenuator (pure-loss) channel followed by a quantum-limited gauge-contravariant channel,
\begin{align}
\tilde{\A} = \tilde{\A}_{\kappa}^0\circ\E_\eta^0~,
\end{align}
with $\tau =\eta(1-\kappa)$ and $y =(\kappa-1)(1-\eta)+\kappa$.
Applying Conjecture~\ref{th:finite-conjecture} for the pure-loss channel $\E_\eta^0$ (which is a gauge-covariant channel) we have
\begin{align}
S(\E_\eta^0(\rho)) \geq S(\E_\eta^0(\rho^{\operatorname{th}}_{g^{-1}(S_0)}))\equiv S_1~.
\end{align}
However, since we proved above that Conjecture \ref{th:finite-conjecture} holds for the quantum-noise-limited gauge-contravariant channel $\tilde{\A}_\kappa^0$
if it holds for single-mode, gauge-covariant bosonic Gaussian channels, we have that
\begin{align}
\nonumber
\min_{S(\rho)=S_0}S(\tilde{A}(\rho))
&=\min_{S(\rho)=S_0}
S(\tilde{\A}_{\kappa}^0\circ\E_\eta^0(\rho))\\
\nonumber
&=\min_{S(\phi)\geq  S_1} S(\tilde{\A}_{\kappa}^0(\phi)) \\
\nonumber
&=S(\tilde{\A}_{\kappa}^0(\phi^{\operatorname{th}}_{g^{-1}(S_1)}))\\
\nonumber
& = S(\tilde{\A}_{\kappa}^0\circ\E_\eta^0(\rho^{\operatorname{th}}_{g^{-1}(S_0)}))\\
& = S(\tilde{A}(\rho^{\operatorname{th}}_{g^{-1}(S_0)}))~.
\end{align}
Therefore we have proved that thermal states minimize the output entropy of all single-mode phase-insensitive (both gauge-covariant and gauge-contravariant) bosonic Gaussian channels, if Conjecture \ref{th:finite-conjecture} is true.

%% file: conclusions.tex
It has been long conjectured that subject to a lower bound on the von Neumann entropy of the input state to a bosonic Gaussian channel (BGC), a thermal state input---one that has a circularly-symmetric, complex Gaussian distribution in phase space---minimizes the entropy at the output of the channel. This conjecture is intimately tied to quantifying the quantum-limited capacities and capacity regions for transmitting classical information over point-to-point and multi-user-broadcast bosonic channels. Various special cases of this conjecture, bounds, and related theorems have been proved in recent years~\cite{giovannetti2014ultimate,mari2014quantum,giovannetti2010generalized,giovannetti2015solution,de2015passive,de2016gaussian}.

In this note, we discuss a proof attempt of the aforesaid Gaussian optimizer conjecture for all single-mode phase-insensitive BGCs (both gauge-covariant and gauge-contravariant). Nearly all practical optical communication channels fall under the class of phase-insensitive BGCs, which is a model general enough to model linear loss (attenuation), additive noise, linear amplification, phase conjugation, and any combinations thereof. As discussed in \cite{QW16}, the truth of Conjecture~\ref{thm:mainresult} implies that trade-off and broadcast capacities of quantum-limited amplifier channels will be solved.

Another open question is the proof of the Gaussian optimizer conjecture for multimode phase-insensitive BGCs, which would imply that with a constraint on the input entropy of the joint state $\rho_{A^n}$ sent over $n$ independent uses of a phase-insensitive BGC ${\cal N}_{A \to B}$, i.e., with $S(\rho_{A^n}) = S_0$, the tensor-product thermal state $\left[\rho^{\rm th}_{\bar n}\right]^{\otimes n}$ minimizes the entropy of the joint $n$-channel-use output state $\rho_{B^n} = \left[{\cal N}_{A \to B}\right]^{\otimes n}\left(\rho_{A^n}\right)$, where ${\bar n} = g^{-1}(S_0/n)$. The proof of the above multi-mode version of Conjecture~\ref{thm:mainresult} is a special case of the entropy photon-number inequality (EPnI)~\cite{guhaEPnI2008,guha2008entropy} and will suffice for all known applications of the EPnI. In particular, the aforesaid multimode version of Conjecture~\ref{th:finite-conjecture} will close the capacity converse proofs of the single-sender multiple-receiver bosonic broadcast channel with loss and thermal noise~\cite{guha2008multiple}, and the triple-tradeoff region of the pure-loss bosonic channel~\cite{2012_wilde_bosonictripletradeoff}. Finally, whether or not the Gaussian-optimizer conjecture holds true for phase-sensitive bosonic Gaussian channels (e.g., a phase-sensitive amplifier, or a squeezing transformation) also remains an open question.

%% file: OneModConj.bbl
\begin{thebibliography}{10}
\providecommand{\url}[1]{#1}
\csname url@samestyle\endcsname
\providecommand{\newblock}{\relax}
\providecommand{\bibinfo}[2]{#2}
\providecommand{\BIBentrySTDinterwordspacing}{\spaceskip=0pt\relax}
\providecommand{\BIBentryALTinterwordstretchfactor}{4}
\providecommand{\BIBentryALTinterwordspacing}{\spaceskip=\fontdimen2\font plus
\BIBentryALTinterwordstretchfactor\fontdimen3\font minus
  \fontdimen4\font\relax}
\providecommand{\BIBforeignlanguage}[2]{{%
\expandafter\ifx\csname l@#1\endcsname\relax
\typeout{** WARNING: IEEEtran.bst: No hyphenation pattern has been}%
\typeout{** loaded for the language `#1'. Using the pattern for}%
\typeout{** the default language instead.}%
\else
\language=\csname l@#1\endcsname
\fi
#2}}
\providecommand{\BIBdecl}{\relax}
\BIBdecl

\bibitem{giovannetti2014ultimate}
V.~Giovannetti, R.~Garc{\'\i}a-Patr{\'o}n, N.~J. Cerf, and A.~S. Holevo,
  ``Ultimate classical communication rates of quantum optical channels,''
  \emph{Nature Photonics}, vol.~8, no.~10, pp. 796--800, 2014.

\bibitem{mari2014quantum}
A.~Mari, V.~Giovannetti, and A.~S. Holevo, ``Quantum state majorization at the
  output of bosonic {Gaussian} channels,'' \emph{Nature Communications},
  vol.~5, p. 3826, 2014.

\bibitem{giovannetti2010generalized}
V.~Giovannetti, A.~S. Holevo, S.~Lloyd, and L.~Maccone, ``Generalized minimal
  output entropy conjecture for one-mode {Gaussian} channels: definitions and
  some exact results,'' \emph{Journal of Physics A: Mathematical and
  Theoretical}, vol.~43, no.~41, p. 415305, 2010.

\bibitem{giovannetti2015solution}
V.~Giovannetti, A.~S. Holevo, and R.~Garc{\'\i}a-Patr{\'o}n, ``A solution of
  {Gaussian} optimizer conjecture for quantum channels,'' \emph{Communications
  in Mathematical Physics}, vol. 334, no.~3, pp. 1553--1571, 2015.

\bibitem{de2015passive}
G.~De~Palma, D.~Trevisan, and V.~Giovannetti, ``Passive states optimize the
  output of bosonic {Gaussian} quantum channels,'' \emph{IEEE Transactions on
  Information Theory}, vol.~62, no.~5, pp. 2895--2906, 2016.

\bibitem{de2016gaussian}
------, ``Gaussian states minimize the output entropy of the one-mode quantum
  attenuator,'' \emph{IEEE Transactions on Information Theory}, vol.~63, no.~1,
  pp. 728--737, January 2017, arXiv:1605.00441.

\bibitem{memarzadeh2016nonGaussian}
L.~Memarzadeh and S.~Mancini, ``Minimum output entropy of a non-{Gaussian}
  quantum channel,'' \emph{Physical Review A}, vol.~94, no.~2, p. 022341,
  August 2016, arXiv:1605.04525.

\bibitem{PTG17}
\BIBentryALTinterwordspacing
G.~De~Palma, D.~Trevisan, and V.~Giovannetti, ``Gaussian states minimize the
  output entropy of one-mode quantum {Gaussian} channels,'' \emph{Physical
  Review Letters}, vol. 118, no.~16, p. 160503, April 2017, arXiv:1610.09970.
  [Online]. Available:
  \url{https://link.aps.org/doi/10.1103/PhysRevLett.118.160503}
\BIBentrySTDinterwordspacing

\bibitem{bell1948shannon}
C.~E. Shannon, ``A mathematical theory of communication,'' \emph{Bell System
  Technical Journal}, vol.~27, pp. 379--423, 1948.

\bibitem{guhaEPnI2008}
S.~Guha, J.~H. Shapiro, and B.~I. Erkmen, ``Capacity of the bosonic wiretap
  channel and the entropy photon-number inequality,'' in \emph{Proceedings of
  the Internaional Symposium on Information Theory (ISIT)}.\hskip 1em plus
  0.5em minus 0.4em\relax IEEE, 2008, pp. 91--95.

\bibitem{guha2008entropy}
S.~Guha, B.~I. Erkmen, and J.~H. Shapiro, ``The entropy photon-number
  inequality and its consequences,'' in \emph{Information Theory and
  Applications Workshop, 2008}.\hskip 1em plus 0.5em minus 0.4em\relax IEEE,
  2008, pp. 128--130.

\bibitem{guha2008multiple}
S.~Guha, ``Multiple-user quantum information theory for optical communication
  channels,'' Ph.D. dissertation, MIT, 2008.

\bibitem{giovannetti2004minimum}
V.~Giovannetti, S.~Guha, S.~Lloyd, L.~Maccone, and J.~H. Shapiro, ``Minimum
  output entropy of bosonic channels: a conjecture,'' \emph{Physical Review A},
  vol.~70, no.~3, p. 032315, 2004.

\bibitem{guha2007classical}
S.~Guha, J.~H. Shapiro, and B.~I. Erkmen, ``Classical capacity of bosonic
  broadcast communication and a minimum output entropy conjecture,''
  \emph{Physical Review A}, vol.~76, no.~3, p. 032303, 2007.

\bibitem{2012_wilde_bosonictripletradeoff}
M.~M. Wilde, P.~Hayden, and S.~Guha, ``Information trade-offs for optical
  quantum communication,'' \emph{Physical Review Letters}, vol. 108, no.
  140501, 2012.

\bibitem{QW16}
H.~Qi and M.~M. Wilde, ``Capacities of quantum amplifier channels,''
  \emph{Physical Review A}, vol.~95, no.~1, p. 012339, January 2017,
  arXiv:1605.04922.

\bibitem{jeff_personal}
J.~H. Shapiro, personal communications, 2014.

\end{thebibliography}
